\documentclass[
 reprint,
superscriptaddress,
 amsmath,amssymb,
 aps,
prx,
]{revtex4-2}

\usepackage{graphicx}
\usepackage{dcolumn}
\usepackage{bm}

\usepackage{amsmath,amssymb,amsfonts,amsbsy,amsthm}
\usepackage{tikz-cd}
\usetikzlibrary{shapes}
\usepackage{amsbsy}
\usepackage{amscd}
\usetikzlibrary{matrix,arrows,decorations.pathmorphing}
\usepackage{graphicx, caption} 
\usepackage{float}
\usepackage{setspace}
\usepackage[margin=1.0in]{geometry}
\usepackage{tikz}
\usepackage{comment}
\usepackage{wrapfig}
\usepackage{hyperref}
\usepackage{enumitem}
\usepackage{nicematrix, bigstrut} 

\usepackage{nicematrix}

\usepackage[justification=raggedright,singlelinecheck=false]{caption}

\newtheorem{theorem}{Theorem}
\newtheorem{proposition}{Proposition}
\newtheorem{lemma}{Lemma}

\theoremstyle{definition}
\newtheorem{definition}{Definition}
\newtheorem*{remark}{Remark}
\newtheorem{example}{Example}

\newcommand{\bra}[1]{\langle #1|}
\newcommand{\ket}[1]{|#1\rangle}

\newcommand{\ip}[2]{\langle #1|#2\rangle}
\newcommand{\op}[2]{|#1\rangle \langle #2|}

\newcommand{\mbf}{\mathbf}
\newcommand{\mbb}{\mathbb}
\newcommand{\mc}{\mathcal}
\newcommand{\msf}{\mathsf}

\newcommand{\tr}{\textrm{Tr}}

\newcommand{\supp}{\text{supp}}

\newcommand{\wt}{\widetilde}
\newcommand{\rk}{\text{rk}}

\newcommand{\Q}{\mathsf{Q}}
\newcommand{\A}{\mathsf{A}}
\newcommand{\B}{\mathsf{B}}

\newcommand{\D}{\mathsf{D}}
\newcommand{\E}{\mathsf{E}}

\newcommand{\rmv}[1]{}

\def\b0{{\bf 0}}

\def\){{\right)}}
\def\({{\left(}}

\newcommand{\R}{\mathsf{R}}

\usepackage{ulem}

\definecolor{cool_green}{rgb}{0.0, 0.5, 0.0}

\newcommand{\sarah}[1]{{\color{blue} #1}}
\newcommand{\todo}[1]{{\color{red} #1}}

\begin{document}

\preprint{APS/123-QED}

\title{Quantum Secret Sharing with a Helper and Programmable Access Structures}

\author{Eric Chitambar}
\email{echitamb@illinois.edu}
 \affiliation{Department of Electrical and Computer Engineering, University of Illinois at Urbana-Champaign}
\author{Sarah Hagen}%
 \email{shagen2@illinois.edu}
\affiliation{%
 Department of Physics, University of Illinois at Urbana-Champaign
}%

\author{David W. Kribs}
 \email{dkribs@uoguelph.ca}
\affiliation{Department of Mathematics \& Statistics, University of Guelph, ON Canada N1G 2W1
}%

\author{Mike I. Nelson}
\email{mnelsonqcr@gmail.com}
\affiliation{Department of Electrical and Computer Engineering, University of Illinois at Urbana-Champaign}%

\author{Andrew Nemec}
\email{andrew.nemec@utdallas.edu}
\affiliation{Department of Computer Science, University of Texas at Dallas
}%

\date{\today}

\begin{abstract}

Quantum secret sharing (QSS) is a process in which the state of a quantum system is partitioned into multiple shares, allowing only specific subsets of shareholders to reconstruct the state, while others gain no information about its identity.
In this work, we propose a variant of QSS, called ``helper QSS,'' which designates a special shareholder whose participation with any other group of parties is sufficient for recovering the state.
We present different families of helper codes that function for an arbitrary number of parties and system sizes.
As an application, we introduce the idea of a programmable access structure in QSS, which allows for a third-party programmer to independently choose the access structure of the code even after the shares are delivered to all the shareholders.  This choice is made \textit{blindly}, meaning that the programmer has no information about the secret being encoded, and it is implemented using the nonlocal process of quantum steering. 

\end{abstract}


\keywords{quantum cryptography, quantum secret sharing, quantum error correction, quantum privacy}

\maketitle

\section{Introduction}

In quantum secret sharing (QSS), parties work together to recover an unknown or secret quantum state that has been encoded into a multipartite system \cite{Cleve-1999a}.
Importantly, only specific subsets of collaborating parties, referred to as authorized sets, can recover the secret.
In contrast, unauthorized sets are groups of parties that collectively have no information about the secret.

This paper initiates the study of QSS schemes in which any subset of parties becomes authorized if it includes a specific designated party, called the helper.
This reflects a scenario in which one individual, such as an elected official, safeguards certain data, and others can obtain it with that person’s assistance.
Alternatively, it could be viewed as a situation in which each of the parties is just one step away from recovering the secret, with the helper serving as the key enabler of access.
This problem may also arise naturally in the development of heterogeneous quantum networks or distributed computing platforms.  
One particular species of qubit, such as trapped ions \cite{Inlek-2017a} or neutral atoms \cite{Covey-2023a}, might be a central node in a network, whereas the other parties could be photonic qubits distributed across lossy channels with easily detectable erasure errors \cite{CY95}. 
In the QSS helper paradigm, the matter-based qubits could serve as the helper, and encoded information could be recovered provided at least one photon arrives at the central node.

Building on the concept of QSS with a helper, we also introduce the notion of a \textit{programmable access structure} (PAS).
The access structure of a QSS specifies which subsets of parties are authorized to recover the secret.
A programmable access structure describes a setting in which the helper is replaced by two different roles: a dealer and a programmer.  The dealer encodes a quantum secret into multiple shares and distributes them, while the programmer decides a particular access structure for the other parties through a judicious choice of local measurement (see Fig. \ref{Fig:PAS_fig}).
The crucial aspect of a PAS is that, unlike in standard QSS, the access structure is decided independent of the encoding, and in fact, it can be decided before the dealer receives the secret or even \textit{after} the shares are delivered to the parties.
This capability might be useful for secure distributed quantum computing, where the level of trust in certain parties could change dynamically.
It could also be leveraged in asynchronous models of computation \cite{Buscemi-2020a, Ji-2024a, Kim-2026a}, where different pieces of classical and quantum information arrive at different times yet still require joint processing.
We show how any QSS helper code can be used to implement a PAS through a straightforward code concatenation.

Our model of a programmable access structure uses just pre-shared entanglement between the programmer and the dealer along with forward classical communication from the programmer to the parties.
This limited type of communication ensures that every PAS is secure against an untrusted programmer.
Since there is no signaling from the dealer to the programmer, no information about the encoded secret is ever leaked to the programmer.
We refer to this as \textit{blind} programming since it is loosely reminiscent of blind quantum computing \cite{Broadbent-2009a}, a computational paradigm in which information processing is driven remotely by some third party who acts on fully encrypted data.
While there is no communication between the programmer and the dealer in a PAS, the quantum correlations they share can be described using the concept of quantum steering \cite{Wiseman-2007a}.  
Interestingly, we show that steering is the key resource underlying every PAS, as quantum steerability between the programmer and dealer is necessary to realize any programmable access structure.

\begin{figure}[t]
\centering
\includegraphics[width=.4\textwidth]{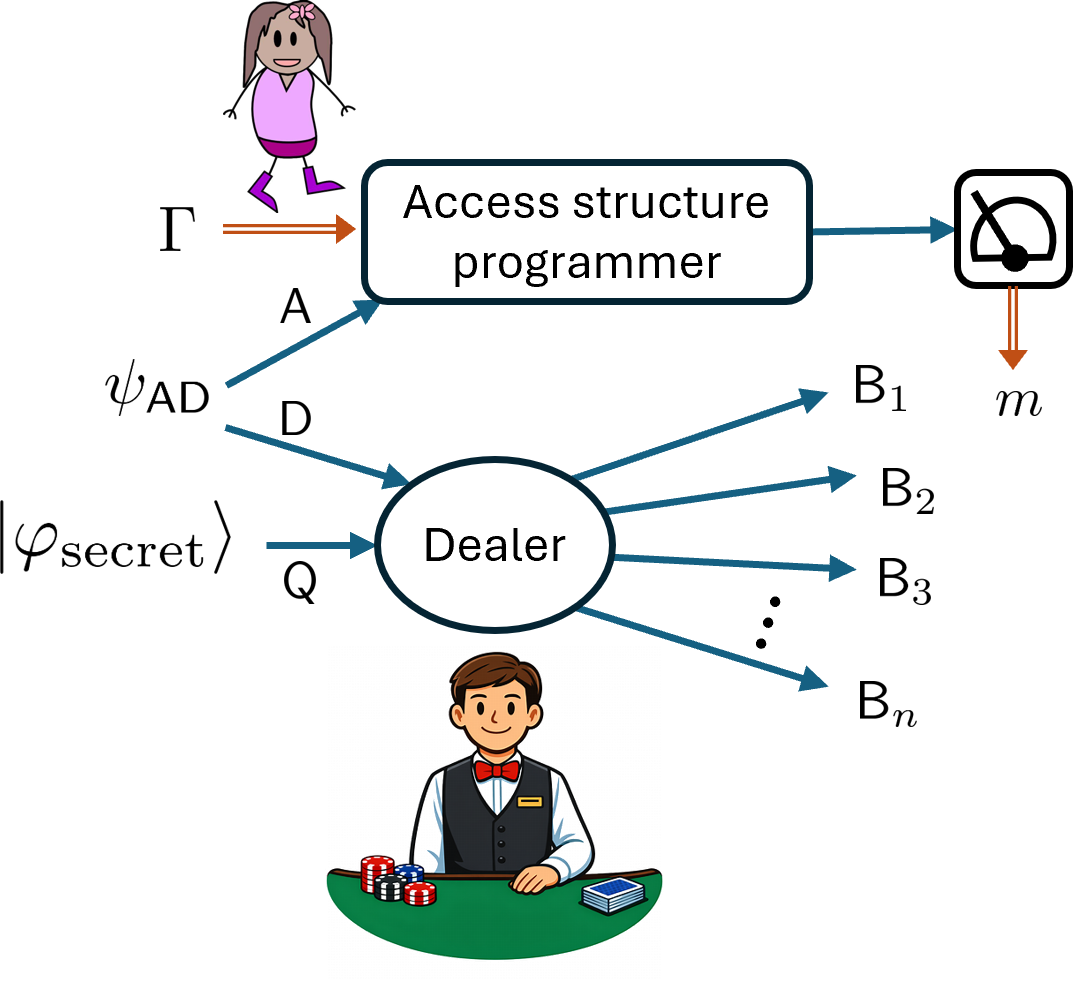}
\caption{A programmable access structure (PAS) involves a dealer who distributes encrypted shares of a quantum secret.  The access structure $\Gamma$ for the $\mathsf{B}_i$ is decided remotely by a programmer at any time before or after the shares are delivered by steering the dealer's quantum system $\D$ using pre-shared entanglement $\varphi_{\A\D}$.  
}
\label{Fig:PAS_fig}
\end{figure}

Historically, quantum secret sharing drew much of its inspiration from the theory of classical secret sharing \cite{Cleve-1999a, Gottesman-2000a, Smith-2000a}.
In particular, Shamir's original work on secret sharing introduced the notion of threshold codes \cite{S79}, which could naturally be extended into the construction of quantum threshold schemes.
In a threshold code, every party is on equal footing, and recovery requires only that a specific membership threshold is attained, irrespective of which parties collaborate.
QSS helper codes can be seen as the opposite extreme in which one party plays a very special role.

Quantum secret sharing beyond threshold schemes has been studied previously in different contexts.
It was realized early on by Gottesman \cite{Gottesman-2000a} and Smith \cite{Smith-2000a} that QSS schemes can be developed with any access structure provided (i) the no-cloning principle is satisfied (meaning it is impossible for two disjoint subsets of parties to recover the secret), and (ii) the access structure is monotonic (meaning that adding parties to a subset can never harm their ability to recover the secret).
Explicit constructions of non-threshold codes have been built using the stabilizer formalism \cite{Markham-2008a, Sarvepalli-2012a, Matsumoto-2019a}, permutation-invariant codes~\cite{sikand2026quantumanonymoussecretsharing}, and other properties of multipartite entanglement \cite{Wang-2013a}.  
An example of a tripartite helper code was presented in Ref. \cite{Jetal20}, and in Section \ref{Sect:Helper-examples} below, we generalize this example to a much larger family of codes.
Ramp secret sharing is a variety of non-threshold secret sharing in which certain subsets of parties are allowed partial  information about the secret \cite{Ogawa-2005a}.
This is in contrast to perfect codes that cleanly separate each subset of parties as either being authorized or unauthorized \cite{Gottesman-2000a}.
Despite the extensive previous work on QSS, relatively little attention has been directed to understanding codes with highly asymmetric access structures; this work aims to help fill this gap.

Our first main contribution is a general framework for studying QSS helper codes based on the structure of quantum error correcting codes for the erasure channel; this is carried out over Sections \ref{sec:preliminaries} and \ref{Sect:Helper-codes}.
In Section \ref{Sect:Helper-examples} we then provide examples of helper codes, which to our knowledge are all novel QSS schemes.
Our constructions offer QSS helper codes for any size secret and any number of parties.
Section \ref{Sect:PAS} introduces the task of access structure programming, and we prove that quantum steering is the key resource that enables blind programming.
As quantum steering is a non-classical effect, our results show that secure programmable access structures are only achievable using quantum resources.
Finally, in Section \ref{Sect:Conclusion} we provide concluding remarks and identify some open problems.

\section{Preliminaries}

\label{sec:preliminaries}

\subsection{Quantum secret sharing}

 A QSS code for state space $\mc{H}_{\mathsf{Q}}$ and parties $\bm{\B}=\mathsf{B}_1,\cdots,\mathsf{B}_n$ is a subspace $\mc{V}\subset\mc{H}_{\bm{\B}}$ of dimension $|\mc{V}|:=\dim(\mc{V})$.
The code is characterized by its access structure $\Gamma\subset 2^{\{\bm{\B}\}}$ whose elements $S\in \Gamma$ are called authorized sets.
Formally, we can describe this condition as follows:
$\forall S\in\Gamma$, there exists a quantum channel $\mc{R}_{S}$ such that \begin{equation}
    \label{Eq:QEC-defn}
\mc{R}_{S}\circ\tr_{\overline{S}}\circ \Pi_{\mc{V}} =\Pi_\mc{V},
        \end{equation}
        where $\tr_{\overline{S}}$ is the partial trace map over systems $\overline{S}:=\{\bm{\B}\}\setminus S$ and $\Pi_{\mc{V}}(\cdot)=P_{\mc{V}}(\cdot)P_\mc{V}$ with $P_{\mc{V}}$ being the projector onto $\mc{V}$.
In this work, we will deal exclusively with pure QSS schemes, which we take to mean a QSS code combined with an isometric encoding map $W:\mc{H}_{\mathsf{Q}}\to \mc{V}$ from some state space $\mc{H}_{\Q}$ into the code subspace $\mc{V}$; i.e. the purity of the encoded state is always preserved.
By channel complementarity (discussed further below), for pure schemes $S \in \Gamma$ implies that $\overline{S} \notin \Gamma$ (but not necessarily the converse). 


The recoverability condition in Eq. \eqref{Eq:QEC-defn} can also be understood in an information-theoretic way following the scheme shown in Fig. \ref{Fig:QSS-circuit}.  
One considers a maximally entangled state $\ket{\Phi^+_{|\Q|}}_{\widetilde{\mathsf{Q}}\mathsf{Q}}=|\mathsf{Q}|^{-1/2}\sum_{i=1}^{|\mathsf{Q}|}\ket{ii}$ shared between $\mathsf{Q}$ and reference system $\wt{\mathsf{Q}}$.  
After the encoder $W$ is applied to $\mc{H}_{\mathsf{Q}}$, the complement of any authorized set $S$ can be discarded.  The recovery map $\mc{R}_{S}$ allows for full restoration of maximal entanglement with $\widetilde{\mathsf{Q}}$.  
It has been known since the early days of quantum error correction \cite{SN96}, that for a fixed authorized set $S$, the process in Fig. \ref{Fig:QSS-circuit} is achievable if and only if $I(\widetilde{\mathsf{Q}}:\overline{S})=0$, where $I$ denotes the quantum mutual information.  
Hence, $W$ is a valid embedding for a QSS scheme with access structure $\Gamma$ if and only if for every  $S\in\Gamma$ and partitioning of the parties $\mc{H}_{\bm{\B}}=\mc{H}_{S}\otimes\mc{H}_{\overline{S}}$, we have $I(\widetilde{\mathsf{Q}}:\overline{S})=0$.

\begin{figure}[t]
\centering
\includegraphics[width=0.4\textwidth]{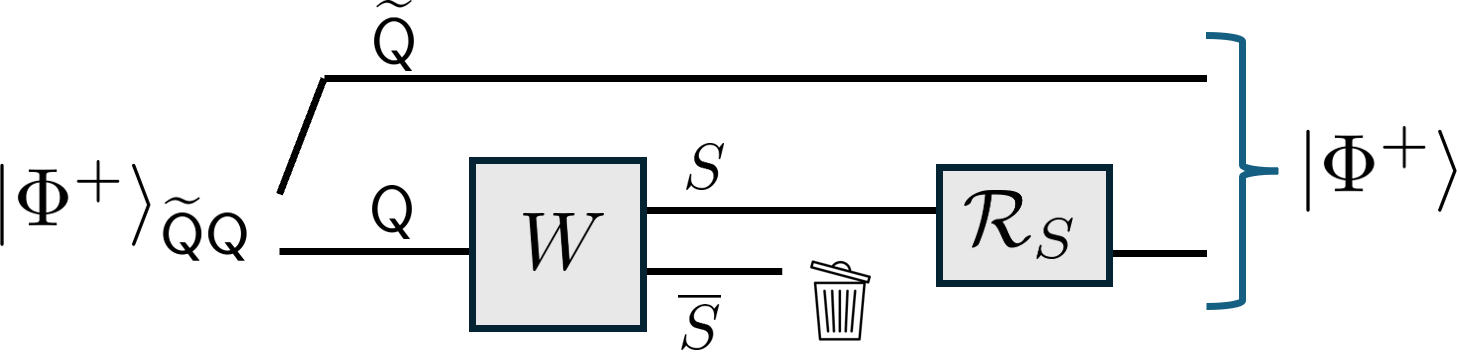}
\caption{For a pure QSS scheme, an encoding isometry $W$ allows for the recovery of maximal entanglement with a reference system using the information contained on any authorized subset $S$ of parties.}
\label{Fig:QSS-circuit}
\end{figure}

This information-theoretic characterization provides a nice connection between recoverability and privacy \cite{kretschmann2008complementarity,crann2016private,kribs2019complementarity}.  
Namely, since $\mathsf{Q}$ and $\wt{\mathsf{Q}}$ begin maximally correlated in the state $\ket{\Phi^+_{|\Q|}}$, we can interpret the condition $I(\widetilde{\mathsf{Q}}:\overline{S})=0$ as roughly saying that full recoverability of the secret by $S$ implies no information about the secret for $\overline{S}$.  
Thus, pure QSS schemes naturally have the feature that the secret is fully private for all $\overline{S}$ whenever $S\in\Gamma$, and conversely privacy for one subset of parties implies recoverability for its complement.
Typically, one considers pure secret sharing schemes that are {\it perfect}, meaning that every $S\subset\{\bm{\B}\}$ satisfies either $S\in\Gamma$ or $\overline{S}\in\Gamma$.
For example, in a $(k,n)$ threshold scheme, any subset of $k$ out of the $n$ parties can recover the secret, while any subset with fewer than $k$ parties holds no information about the secret.
An alternative to perfect secret sharing is ramp secret sharing \cite{Ogawa-2005a}, which allows for some subsets $S\not\in\Gamma$ for which $I(\wt{Q}:S)\not=0$.

\subsection{The structure of erasure codes}

Regardless of whether one considers perfect or ramped secret sharing, a necessary condition for every scheme is the ability to correct the erasure of subsystems $\overline{S}$ for every $S\in\Gamma$.  
Hence, the theory of quantum erasure channels is fundamentally connected to the theory of quantum secret sharing.  
We can understand the quantum erasure channel $\mc{E}_{S}$ for parties $S$ as simply the tracing out of parties $\overline{S}$, i.e. $\mc{E}_{S}=\tr_{\overline{S}}$.  
As shown in Fig. \ref{Fig:QSS-circuit}, the isometry $W:\mc{H}_\mathsf{Q}\to \mc{V}$ serves as an encoding that protects against the noisy channel $\mc{E}_{S}$.  
In fact, for the purpose of QSS, $\mc{V}$ must be a code for the entire collection of erasure channels $\{\mc{E}_{S}\}_{S\in\Gamma}$.

In recent work, we have characterized all codes that can correct a single erasure channel $\mc{E}_{S}=\text{Tr}_{\overline{S}}$.
\begin{lemma}[\cite{Cetal25}]\label{lem:equiv-cond}
    Given the setup of Fig. \ref{Fig:QSS-circuit}, the following are equivalent.
    \begin{enumerate}
    \item[\upshape{(a)}] A code $\mc{V}\subset\mc{H}_{S}\otimes\mc{H}_{\overline{S}}$ can correct the erasure $\mc{E}_{S}$; i.e., there is a map $\mc{R}_{S}$ such that Eq. \eqref{Eq:QEC-defn} holds. 
        \item[\upshape{(b)}] $I(\wt{\Q}:\overline{S})_{\mbb{I}_{\wt{\Q}}\otimes W\ket{\Phi^+}_{\wt{\Q}\Q}}=0$.
       \item[\upshape{(c)}] $\mc{V}$ has an orthonormal basis $\{\ket{i}_{S\overline{S}}\}$ of the form
\begin{align}\label{eq:codeword-form}
\ket{i}_{S\overline{S}}&=\sum_{k=1}^{\rk(\overline{S})}\sqrt{p_k}\ket{\varphi_{i,k}}_{S}\ket{\gamma_k}_{\overline{S}},
\end{align}
where the $\{\ket{\gamma_k}\}_k$ form an orthonormal set and the set of states $\{\ket{\varphi_{i,k}}_{S}\}_{i,k}$ satisfy 
 $\ip{\varphi_{i,k}}{\varphi_{j,l}}=\delta_{kl}\delta_{ij}$.  
   Here, $\rk(\overline{S})$ refers to the rank of systems $\overline{S}$ in the state $\mbb{I}_{\wt{\Q}}\otimes W\ket{\Phi^+}_{\wt{\Q}\Q}$.
        \item[\upshape{(d)}] There exists an auxiliary system $\mathsf{R}$, a fixed state $\ket{\psi}_{\Q\overline{S}}$, and isometry $U_{\Q\R\to S}$ such that
        \begin{align}
    \ket{i}_{S\overline{S}}&=(U_{\Q\R\to S}\otimes\mbb{I}_{\overline{S}})\ket{i}_{\Q}\otimes\ket{\psi}_{\R\overline{S}}\;\;\forall i=1,\ldots , |\mc{V}|.\notag
\end{align}
\end{enumerate}
\end{lemma}
\noindent Condiiton (c) was not explicitly included in \cite{Cetal25}, but it is easily seen to be equivalent to the code state condition (d). Note that for a QSS scheme, these conditions must apply for any bipartitioning $\bm{\B}=S\overline{S}$ with $S\in \Gamma$.  


\section{Helper quantum secret sharing}

\subsection{Helper codes}

\label{Sect:Helper-codes}

Now suppose we consider an additional system $\mathsf{A}$ so that there are $n+1$ total systems labeled $\A, \B_1, \cdots, \B_n$.
We envision $\A$ as a helper that can assist any of the $\B_i$ in decoding the encoded quantum information.
This leads to the notion of a QSS helper scheme.
\begin{definition}
    A quantum secret sharing code for parties $\A,\B_1,\cdots,\B_n$ is called a QSS helper code with helper $\mathsf{A}$ if its access structure $\Gamma$ satisfies 
\begin{enumerate}
\item[(i)] $\mathsf{A}\in S \implies S\in\Gamma$;
\item[(ii)] $\{\mathsf{A}\}\not\in \Gamma$.
\end{enumerate}
We call the QSS code a \textit{blind} helper code if, in addition,
\begin{enumerate}
\item[(iii)] $\{\bm{\B}\}\in \Gamma$, 
\end{enumerate}
where  $\bm{\B}=\msf{B}_1,\cdots,\msf{B}_n$.

\end{definition}

\medskip

\noindent Condition (i) says that the mere inclusion of the helper $\mathsf{A}$ in any cohort of parties is sufficient to recover the secret.  
By monotonicity of access structures, we can equivalently reduce this to the condition that $(\mathsf{A},\mathsf{B}_1),(\mathsf{A},\mathsf{B}_2),\cdots,(\mathsf{A},\mathsf{B}_n)\in\Gamma$.  
Condition (ii) is a non-triviality condition, meaning that Alice (party $\mathsf{A}$) cannot recover the secret on her own.  
This allows for some level of safeguarding since at least one other party $\mathsf{B}_i$ must collaborate with her. 

Even more fundamentally, if Alice is able to recover the secret acting by herself, then the problem is equivalent to correcting for the erasure of $\bm{\B}$, and Lemma \ref{lem:equiv-cond} already provides a full characterization of such codes.  That is, $\{\ket{i}_{\mathsf{A}\bm{\B}}\}$ provides a basis for code $\mc{V}$ if and only if we can write $\ket{i}_{\mathsf{A}\bm{\B}}=U\ket{i}_{\mathsf{A}_1}\ket{\psi}_{\mathsf{A}_2\bm{\B}}$ for some local isometry $U:\mathsf{A}_1\mathsf{A}_2\to\mathsf{A}$ on Alice's side and some fixed shared state $\ket{\psi}_{\mathsf{A}_2\bm{\B}}$.
Condition (ii) rules out such trivial codes.

When condition (iii) holds, which is a stronger condition than (ii) (i.e., (iii) implies (ii)), we describe the code as being blind due to complementarity.
Again following Lemma \ref{lem:equiv-cond}, the recovery of system $\mathsf{Q}$ by $\bm{\B}$ implies that $I(\wt{Q}:\mathsf{A})=0$. 
In other words, Alice is completely uncorrelated with the reference system $\wt{Q}$.
Operationally, this can be understood as a statement of privacy since her system will then also be uncoupled to the secret recovered by $\bm{\B}$.

Every pure QSS helper scheme is \textit{almost} a perfect scheme.  
This is because every subset $S\not=\{\bm{\B}\}$ of the parties is either authorized or it does not contain $\mathsf{A}$; in the latter case, $\overline{S}$ will contain $\mathsf{A}$ (and at least one other $\mathsf{B}_i$) and is therefore authorized.
For non-blind QSS schemes the full set of Bobs $\bm{\B}$ is not authorized even though $\overline{\{\bm{\B}\}}=\{\mathsf{A}\}\not\in\Gamma$; hence these are not perfect and $I(\wt{Q}:\A)\not=0$.
In contrast, blind QSS helper codes explicitly require that $\{\bm{\B}\}\in\Gamma$, and so these codes are indeed perfect.

Blind QSS codes are a very natural class of codes to consider since they permit recovery of the secret under full collaboration of all the parties $\B_1,\cdots,\B_n$.
This provides a balance of power so that assistance of the helper is sufficient for recovery but not necessary.
It also restricts the helper to be an impartial participant, one who can just enable recovery without acquiring any information independently.
Also, from an implementation perspective, blind QSS codes have the advantage that the encoding can be accomplished by acting non-trivially only on the $\bm{\B}$ systems.
That is, condition (d) of Lemma \ref{lem:equiv-cond} implies that blind QSS codes have a logical basis of the form
\begin{align}
    \label{Eq:blind-encoding}
\ket{i}_{\A\bm{\B}}=\mathbb{I}_{\A}\otimes U_{\R\Q\to\bm{\B}}\ket{\psi}_{\A\R}\ket{i}_{\Q}.
\end{align}
Also, as we show in Example \ref{Ex:universal-blind} below, every non-blind code can be converted into a blind one 
through the use of auxiliary Bell states.

\subsection{Examples of helper codes}

\label{Sect:Helper-examples}

We now provide some examples of QSS helper codes.

\begin{example}[Helper secret sharing from threshold secret sharing]\label{ex:threshold-helper-codes}
    Our first class of examples involves reassigning shares of traditional threshold secret sharing schemes.
In the seminal work of \cite{CGL99}, the authors touch on a method to build more general access secret sharing schemes beyond thresholds schemes.
    Specifically, any $(n,k)$ threshold secret sharing scheme can be transformed into a helper secret sharing scheme by giving $k-1$ shares to the helper.
    The helper's cooperation with any other party now becomes sufficient in recovering the secret.
    In some cases, the collection of non-helper participants also forms an authorized set and the helper is blind.
    For odd $n$, when $k\neq\frac{n+1}{2}$, the helper is not blind and the collection of non-helper participants does not form an authorized set.
For any secret dimension $|\Q|$, it is always possible to choose a qudit dimension $d$, a number of qudits $n$ and a threshold $k$, such that an integer number of qudits can be reassigned to the helper \cite{CGL99}. 
However, in this case, the resulting dimension of the helper system generally grows faster than the code dimension. 
    %
\end{example}

When trying to construct helper codes, an interesting question is whether there is any relationship between the size of the encoded state space $\mc{H}_{\Q}$ and the system dimension of the shareholders.
One relatively simple bound on $|\Q|$ can be stated in terms of the system dimensions.
\begin{proposition}
\label{Prop:Helper-dim-bound}
    Every pure helper scheme of $\Q$ for systems $\A, \B_1, \cdots  \B_N$ with $\A$ being the helper satisfies
    \begin{equation}
    \label{Eq:helper-dim-ineq}
       |\Q|\leq \min\{\frac{1}{2}|\mathsf{A}|\cdot \min_{i}|\mathsf{B}_i|,\;|\mathsf{A}|\}.
    \end{equation}
\end{proposition}
\begin{proof}
The bound $|\mathsf{A}|\geq|\mathsf{Q}|$ follows from entropic arguments in Refs. \cite{Hetal05, Gottesman-2000a}.
In the terminology of \cite{Leung-2010a}, any helper system $\msf{A}$ is ``significant'' in the sense that, say, $\B_2\cdots\B_N$ have no information about the secret (since $\A\B_1$ is authorized), and yet adding $\A$ to the cohort makes $\A\B_2\cdots\B_N$ authorized.
For any significant party, it has been shown that
\begin{align}
    \log|\Q|=S(\wt{Q})\leq S(\A)\leq\log |\A|,
\end{align}
where $S(\cdot)$ denotes von Neumann entropy.

To prove the inequality $\frac{1}{2}|\A|\cdot\min_i|\B_i|\geq |\Q|$,
relabel systems such that $\mathsf{B}_1$ satisfies $|\mathsf{B}_1|=\min_{i}|\mathsf{B}_i|$.  By Lemma \ref{lem:equiv-cond}, $(\mathsf{A},\mathsf{B}_1)$ being an authorized set means that code space $\mc{V}\cong\mc{H}_{\Q}$ has an orthonormal basis $\{\ket{i}\}$ such that
    \begin{align}
        \ket{i}_{\A \B_1\cdots  \B_N}=\sum_{k=1}^{\rk(\rho_{\B_2\cdots  \B_N})}\sqrt{p_{k}}\ket{\varphi_{i,k}}_{\A\B_1}\ket{\gamma_k}_{\B_2 \cdots  \B_N}.\notag
    \end{align}
    Orthogonality of the $\{\ket{\varphi_{i,k}}\}$ immediately implies that $|\Q|\cdot \rk(\rho_{\B_2, \cdots  \B_N})\leq |\A|\cdot|\B_1|$.
    The proposition is then proven by noting that $\rk(\rho_{\B_2 \cdots  \B_N})\geq 2$.
    Indeed, otherwise systems $\B_2 \cdots  \B_N$ would be decoupled from $\wt{\mathsf{Q}}\A\B_1$ in Fig. \ref{Fig:QSS-circuit}, which means that $I(\wt{\mathsf{Q}}:\B_1\B_2\cdots\B_N)=I(\wt{\mathsf{Q}}:\B_1)$=0, with the last equality following from the fact that full recovery is possible when $\B_1$ is discarded.
    However, this is not possible since $I(\wt{\mathsf{Q}}:\B_1\B_2\cdots\B_N)=0$ would mean that Alice could recover the secret entirely on her own. 
\end{proof}

The following example provides an instance of when the inequality in Eq. \eqref{Eq:helper-dim-ineq} is tight.  
\begin{example}[Tripartite helper code encoding arbitrary dimensions]
\label{Ex:FormI}
    Let $\mc{V}\cong\mc{H}_{\Q}$ be the code with a $|\Q|$-dimensional helper $\A$, $d$-dimensional parties $\B_1,\B_2$, and orthonormal basis  
    \begin{align}
    \label{Eq:FormId-eq1}
           \ket{i}&=\ket{i}_\A\frac{1}{\sqrt{d}}\sum_{k=1}^{d}(U_i\otimes\mathbb{I})\ket{kk}_{\B_1\B_2},\quad 1\leq i\leq |\Q|,
    \end{align}
    where the $U_i$ are arbitrary unitaries.
    Note that the basis states can be equivalently written as
    \begin{align}
    \label{Eq:FormId-eq2}
         \ket{i}&=\ket{i}_\A\frac{1}{\sqrt{d}}\sum_{k=1}^{d}(\mathbb{I}\otimes U_i^{\mathtt{T}})\ket{kk}_{\B_1\B_2},
    \end{align}
    where $U_i^{\mathtt{T}}$ denotes the matrix transpose.
    The recoverability is now transparent.
    Parties $\A\B_1$ can recover by performing a controlled $U^\dagger_i$ gate, i.e. $\sum_{i=1}^{d}\op{i}{i}\otimes U^\dagger_i$, while for $\A\B_2$ the recovery is facilitated by a controlled $\overline{U}_i\equiv(U^\dagger_i)^{\mathtt{T}}$ gate.
Note that this code saturates the bound of Eq. \eqref{Eq:helper-dim-ineq} when $d=2$, for an arbitrary code size $|\Q|$.
\end{example}

The next example extends the previous one to $n+1$ parties.  
However, when introducing more parties, an equivalence like Eqns. \eqref{Eq:FormId-eq1} and \eqref{Eq:FormId-eq2} will only hold for diagonal unitaries $U_i$.
\begin{example}[($n+1$)-partite helper code encoding arbitrary dimensions]
\label{Ex:FormII}
    Let $\mc{V}\cong\mc{H}_{\Q}$ be the code with a $|\Q|$-dimensional helper $\A$, $d$-dimensional parties $\bm{\B}=\B_1,\cdots \B_n$, and orthonormal basis  
    \begin{align}
        \ket{i}&=\ket{i}_\A \sum_{k=1}^{d}e^{\mathrm{i}
        \phi_{i,k}}\sqrt{p_k}\ket{k\cdots k}_{\bm{\B}},\quad1\leq i\leq |\Q|,
    \end{align}
   for arbitrary phases $e^{i\phi_{i,k}}$ and probabilistic weights $p_k$.
   Given system $\A$, party $\B_i$ can decouple from all other parties by the unitary $\sum_{i,k=1}^{d}e^{-i\phi_{i,k}}\op{ik}{ik}$. 
\end{example}

\begin{remark}
\label{ex:vip-code}
Examples \ref{Ex:FormI} and \ref{Ex:FormII} generalize the tripartite ``VIP code'' from \cite{Jetal20}, which has encoded basis states of the form
    \begin{align}
    \label{Eq:vip-code} 
        \ket{i}=\sum_{j,k,l=0}^{d-1}\frac{1}{d}t_{ijkl}\ket{jkl}_{\A\B_1\B_2},\quad 0\leq i\leq d-1,
    \end{align}
    where $t_{ijkl}=w^{ij}\ip{k+l}{i}$ and $w=e^{2\pi\mathrm{i}/d}$.
To see the connection, introduce the discrete Fourier transform $H=\frac{1}{\sqrt{d}}\sum_{m,n=0}^{d-1}\omega^{mn}\op{m}{n}$.  Then one computes
\begin{align}
    H\otimes\mbb{I}\otimes\mbb{I}\ket{i}&=\ket{i}_{\A}\frac{1}{\sqrt{d}}\sum_{k,l=0}^{d-1}\ip{k+l}{i}\ket{k}_{\B_1}\ket{l}_{\B_2}\notag\\
    &=\ket{i}_{\A}\frac{1}{d^{3/2}}\sum_{j,k,l=0}^{d-1}\omega^{j(k+l-i)}\ket{k}_{\B_1}\ket{l}_{\B_2}\notag\\
    &=\ket{i}_{\A}\frac{1}{\sqrt{d}}\sum_{j=0}^{d-1}\omega^{-ij}H\ket{j}_{\B_1}\otimes H\ket{j}_{\B_2}.
\end{align}
Thus, up to the local unitaries $H\otimes H\otimes H$, the code of Eq. \eqref{Eq:vip-code} has the form of Examples \ref{Ex:FormI} and \ref{Ex:FormII} above for $|\Q|=|\A|=|\B_i|=d$.
\end{remark}

Neither Example \ref{Ex:FormI} nor \ref{Ex:FormII} is a blind QSS helper code.  
The next example describes a code for which the helper is indeed blind.

\begin{example}[Teleportation codes]

A simple blind helper code can be built through a coherent teleportation protocol.
The scheme is given in Fig. \ref{Fig:QSS_teleportation_code}.
For input system $\Q$, suppose that the helper's system $\A$ is $|\Q|$-dimensional and in the maximally entangled state $\ket{\Phi^+_{|\Q|}}_{\A\R}$ with a reference system $\R$.
Standard teleportation involves performing a $|\Q|$-dimensional Bell measurement on systems $\Q\R$, thereby transferring the information in system $\Q$ to system $\A$, up to some Pauli correction \cite{Bennett-1993a}.
Operationally, the Bell measurement can be carried out by rotating the Bell basis to the computational basis and then measuring systems $\Q\R$ in the computational basis.
Our encoding scheme, however, maintains coherence by skipping the measurement.
Instead, a coherent broadcasting of the computational basis is performed, \[U=\sum_{k=1}^{|\Q|^2}\ket{k}_{\bm{\B}''}^{\otimes n}\bra{k},
\]
where $\bm{\B}''=\B_1''\cdots \B_n''$ are $n$ auxiliary systems delivered to the $n$ parties $\bm{\B}$.
We call the unitary circuit carrying out this basis rotation and the broadcasting $U$ as a ``coherent Bell measurement (BM).'' 
Note that $k$ is an index for the $2\log_2|\Q|$ different measurement outcomes that would be obtained if the Bell measurement occurred.  
Additionally, the helper and each of the $\B_i$ share a maximally entangled state $\ket{\Phi^+_{|\Q|}}_{\A\B_i}$.

\begin{figure}[t]
\centering
\includegraphics[width=0.45\textwidth]{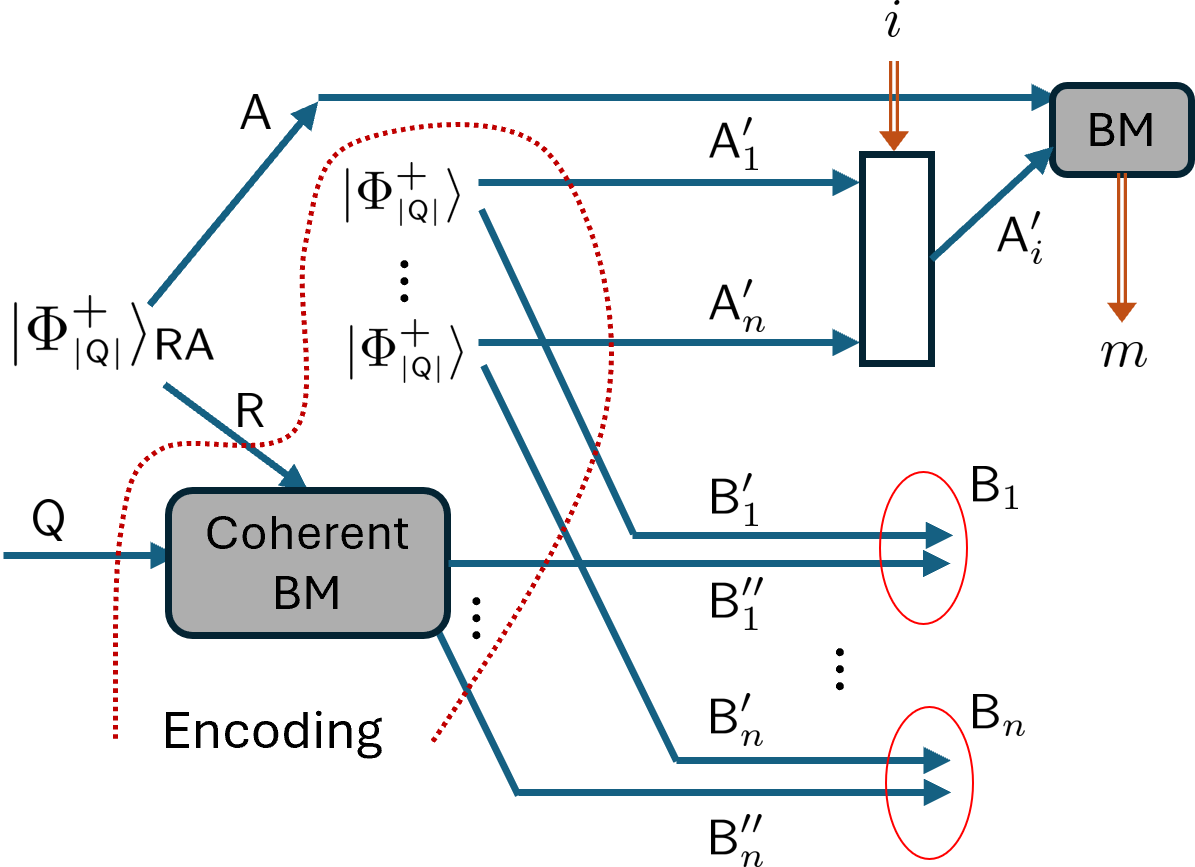}
\caption{Teleportation codes.  A blind QSS code can be constructed by combining a coherent Bell measurement (BM) with a teleportation from the helper to the desired party.  The coherent Bell measurement can be reversed, which allows for the $\bm{\B}$ to collectively recover the secret, while any individual $\B_i$ can recover the secret by effectively triggering a teleportation of the secret to the helper and then having it returned by another round of teleportation.}
\label{Fig:QSS_teleportation_code}
\end{figure}

We now describe how the secret can be recovered.  
Suppose that party $\B_i$ wishes to learn the secret with the help of party $\A$.
Party $\B_i$ simply measures system $\B_i''$ but does not announce the measurement outcome, which effectively teleports the secret to $\A$ up to a Pauli error. 
The helper then teleports the secret back to $\B_i$ by performing a Bell measurement on systems $\A\A_i'$ and communicating the outcome.
The secret can be recovered on system $\B_i'$ by applying the required Pauli corrections from the two teleportation measurements.
Note that $\A$ remains blind during this process since $\B_i$ never reveals the Pauli error after the first teleportation.
On the other hand, the full cohort of parties $\bm{\B}$ can recover the secret without the help of $\A$ by collectively performing $U^\dagger$ on systems $\bm{\B}''$ and then rotating the computational basis back into the Bell basis. 
\end{example}

\begin{example}[Making any helper code blind]
\label{Ex:universal-blind}
The general ideas employed in the previous example can be used to transform any helper code into one that is blind at the cost of needing more Bell pairs.
Let $W:\mc{H}_{\Q}\to\mc{H}_{\A'}\otimes\mc{H}_{\bm{\B}}$ be an encoder for any QSS helper code.
To make this blind, introduce an additionally maximally entangled state $\ket{\Phi^+_{|\A'|}}_{\A\R}$ with system $\A$ belonging to the helper and having dimension $|\A'|=|\A|$.
We then perform a coherent Bell measurement on systems $\A'\R$, as in the teleportation code, which outputs auxiliary systems $\B_1'',\cdots,\B_n''$.
When combining $W$ with the coherent Bell measurement, we have an encoder for a blind helper code for systems $\A,(\B_1\B_1''),\cdots,(\B_n\B_n'')$.

In addition, we can augment the code with extra teleportation so that the helper only needs to communicate classically.
For each party $\B_i$, the helper needs to share an additional state $\ket{\Phi_{|\A|}^+}_{\A'_i\B_i'}$; the helper can then teleport $\A$ to any of the $\B_i$ and enable recovery after communicating classically the necessary Pauli correction. 
\end{example}

Our final example is a stabilizer code that is the most resource-efficient blind helper code among all the ones discussed in this paper.

\begin{example}[CSS helper codes]\label{ex:CSS-code-construction}

The codes we present here have the form of Calderbank–Shor–Steane (CSS) stabilizer codes.
They are blind helper codes that distribute one qubit to an arbitrary number of parties.
Compared to the teleportation codes described above, they are more efficient in terms of the size of the shares.
They also possess the salient feature that the helper enables recovery through forward classical communication alone.

Let $I_{X}^n$ (resp. $I_{Z}^n$) denote the $n\times n$ operator matrix that has $X$ (resp. $Z$) along the diagonal and $I$ everywhere else.  Let $J_{X}^n$ (resp. $J_Z^n$) denote the $n\times n$ operator matrix that has $I$ along the anti-diagonal and $X$ (resp. $Z$) everywhere else.  Finally, let $\mbf{c}^n_X$ (resp. $\mbf{c}^n_Z$) denote the one-column operator matrix with each element being $X$ (resp. $Z$).

For any odd integer $n\geq 3$, we construct a $[[2n+1,1]]$ stabilizer code in which $n$ parties $\B_1,\cdots,\B_n$ each hold a qubit, and the helper $\A$ holds $n+1$ qubits.
We can handle the case of even $n$ by starting with an odd number of parties and then grouping two of them together.
Our codes are CSS codes with stabilizer generators given by the rows of the matrix
\begin{align}
   \begin{pNiceArray}{cc:c}
I_X^n&\mbf{c}^n_X&J_X^n \\
I_{Z}^n&\mbf{c}^n_Z& J_Z^n
    \end{pNiceArray}_{2n\times 2n+1}
\end{align}
Here, the qubits are ordered such that the first $n+1$ columns (to the left of the dash line) are held by helper, while the remaining $n$ qubits (to the right of the dashed line) are held by the other parties.
Overall, we have $2n$ independent and pairwise commuting Pauli strings, which define a one-qubit stabilizer code.
Logical operators are given by
\begin{align}
    \overline{Z}&=I\cdots I:ZZ\cdots Z\notag\\
    \overline{X}&=I\cdots I:XX\cdots X,
\end{align}
which act trivially on the helper's qubits.
Two examples of our codes are 
\begin{align}
    n=3:\;\; \begin{pNiceArray}{cccc:ccc}
X&I&I&X&X&X&I\\
I&X&I&X&X&I&X\\
I&I&X&X&I&X&X\\
Z&I&I&Z&Z&Z&I\\
I&Z&I&Z&Z&I&Z\\
I&I&Z&Z&I&Z&Z
    \end{pNiceArray},\notag
\end{align}
which is precisely the Steane code \cite{St96}, and
\begin{align}
    n=5:\;\; \begin{pNiceArray}{cccccc:ccccc}
X&I&I&I&I&X&X&X&X&X&I\\
I&X&I&I&I&X&X&X&X&I&X\\
I&I&X&I&I&X&X&X&I&X&X\\
I&I&I&X&I&X&X&I&X&X&X\\
I&I&I&I&X&X&I&X&X&X&X\\
Z&I&I&I&I&Z&Z&Z&Z&Z&I\\
I&Z&I&I&I&Z&Z&Z&Z&I&Z\\
I&I&Z&I&I&Z&Z&Z&I&Z&Z\\
I&I&I&Z&I&Z&Z&I&Z&Z&Z\\
I&I&I&I&Z&Z&I&Z&Z&Z&Z\\
    \end{pNiceArray}.\notag
\end{align}

The decoding procedure for our codes is straightforward.  Suppose that Alice wants to localize the secret to party $\B_i$ with $i\in\{1,2,\cdots,n\}$.  She measures $Z$ on both her qubits $n+1$ and $n+1-i$, while she measures $X$ on all of her other qubits.  To analyze the action of these measurements, suppose without loss of generality that $i=1$.  The measurement of $X_1,X_2,\cdots,Z_n$ on her qubits yields the state with stabilizer generators
\begin{align}
    \begin{pNiceArray}{c:c}
\begin{array}{c}
     (-1)^{a_1}X\\ 
      \vdots\\
      (-1)^{a_{n-1}}X
\end{array}&J_X^{n-1}\\
(-1)^{a_n}Z&I\;Z\;\cdots\; Z
    \end{pNiceArray},
\end{align}
where the values of $a_1,\cdots,a_n$ correspond to Alice's outcomes.
Equivalently, by adding row $1$ to rows $2,3,\cdots,n-1$, we can express the generators as
\begin{align}
   \begin{pNiceArray}{c:ccccccc}
   X&X&X&X&\cdots &X&X&(-1)^{a_{1}}I\\
   I&I&I&I&\cdots &I&X&(-1)^{a_{2}}X\\
   I&I&I&I&\cdots &X&I&(-1)^{a_{3}}X\\
       &&\vdots&&&&\vdots&\\
   I&I&X&I&\cdots &I&I &(-1)^{a_{n-1}}X\\
   Z&I&Z&Z&\cdots &Z&Z &(-1)^{a_n}Z
    \end{pNiceArray}.
    \end{align}
Then when Alice measures $Z$ on her remaining qubit, the final state for the other parties has stabilizer generators
\begin{align}
   \begin{pNiceArray}{ccccccc}
   I&I&I&\cdots &I&X&(-1)^{a_{2}}X\\
   I&I&I&\cdots &X&I&(-1)^{a_{3}}X\\
       &\vdots&&&&\vdots&\\
   I&X&I&\cdots &I&I &(-1)^{a_{n-1}}X\\
  I&Z&Z&\cdots &Z&Z &(-1)^{a_n+a_{n+1}}Z
    \end{pNiceArray}.
    \end{align}
These stabilize the one-qubit space $\B_1$, with systems $\B_2\cdots\B_n$ fully decoupled. 
Moreover, since $\overline{Z}=I\cdots I:ZZ\cdots Z$ and $\overline{X}=I\cdots I:XX\cdots X$, we see that the logical operators transform as
\begin{align}
    \overline{Z}&\mapsto (-1)^{a_n+a_{n+1}}Z_{\B_1}\notag\\
    \overline{X}&\mapsto (-1)^{a_2+a_3+\cdots a_{n-1}}X_{\B_1}.
\end{align}
Hence, all the quantum information has been localized onto qubit $\B_1$, as desired.


Since the logical operators $\overline{Z}$ and $\overline{X}$ act trivially on the helper's system, the code is blind \cite{Bravyi-2009a} (see also Ref. \cite{Companion-temp}).
Hence, the encoding can be accomplished using pre-shared entanglement $\ket{\psi}_{\R\A}$ as in Eq. \eqref{Eq:blind-encoding}.
Using the methods of Ref. \cite{Fetal04}, one can show that it suffices for the helper to share $n-1$ ebits with the reference system in the state $\ket{\psi}$.
This should be compared with the teleportation codes, in which $n$ ebits are needed.  
Furthermore, the shares held by parties $\B_i$ are twice as large in the teleportation code for $d=2$ relative to the CSS codes. 

\end{example}

\section{Programmable Access Structures}
\label{Sect:PAS}

\subsection{Defining a PAS}

The idea of a programmable access structure (PAS) is shown in Fig. \ref{Fig:PAS_fig}.
The general scenario consists of a dealer who performs an entanglement-assisted encoding of quantum system $\mathsf{Q}$ via some completely-positive trace-preserving (CPTP) map $\mc{E}_{\mathsf{DQ}\to\bm{\B}}:\mathrm{L}(\mc{H}_{\mathsf{D}}\otimes\mc{H}_{\msf{Q}})\to \mathrm{L}(\mc{H}_{\bm{\B}})$, which acts on the set of linear operators for joint state space $\mc{H}_\D\otimes\mc{H}_\Q$.  Here, $\mathsf{D}$ is an auxiliary system initially entangled with the programmer Alice in the state $\psi_{\mathsf{AD}}\in\mathrm{L}(\mc{H}_{\A}\otimes\mc{H}_{\D})$ (which may or may not be pure), and  $\bm{\B}=\msf{B}_1,\cdots,\msf{B}_n$ is a collection of $n$ parties.
Associated with this embedding is a family of access structures $\mathfrak{A}=\{\Gamma_i\}_i$ with $\Gamma_i\subset 2^{\{\bm{\B}\}}$ and a measurement assemblage $\{M_{a|i}\}_{a,i}$ for the programmer (i.e. a family of positive operator-valued measures (POVMs) for system $\A$ such that the $i^{th}$ POVM has effects $M_{a|i}\geq 0$ satisfying $\sum_{a}M_{a|i}=\mathbb{I}_\A$).
In the following definition, we let $\text{supp}\{X\}$ denote the support of positive operator $X$ and $\mathbb{I}_{\mathsf{Q}}$ the identity operator on $\mc{H}_{\mathsf{Q}}$.
\begin{definition}
\label{Defn:PAS}
A programmable access structure (PAS) for input system $\Q$ and parties $\bm{\B}$ is defined by a tuple $(\mathfrak{A},\;\psi_{\mathsf{AD}},\;\mc{E}_{\mathsf{DQ}\to\bm{\B}},\;\{M_{a|i}\})$ 
with $\mathfrak{A}=\{\Gamma_i\}\subset 2^{\{\bm{\B}\}}$ being a family of access structures such that the following two conditions hold: \\ \noindent (i)
    \begin{align}
\label{Eq:PAS}
   \text{supp}\!\left\{\!\tr_{\mathsf{A}}[ (M_{a|i}\!\otimes\!\mathbb{I}_{\bm{\B}} ) \text{id}_{\A}\otimes\mc{E}_{\D\Q\to\bm{\B}}(\psi_{\A\D}\otimes \mathbb{I}_{\mathsf{Q}}  \!)]\!\right\}
\end{align}
is a QSS code for $\Q$ (denoted by $\mc{V}_{a|i}$) with access structure $\Gamma_i$ for every outcome $a$ and every choice of access structure $\Gamma_i\in\mathfrak{A}$; and \\ \noindent (ii)
 \begin{align}
\label{Eq:PAS-joint}
   \mc{V}=\text{supp}\left\{ \text{id}_{\A}\otimes\mc{E}_{\D\Q\to\bm{\B}}(\psi_{\A\D}\otimes \mathbb{I}_{\mathsf{Q}}  \!)\right\}
\end{align}
is a code for $\Q$ that can correct the erasure of $\A$.
Any tuple satisfying this definition is said to realize a programming of access structure set $\mathfrak{A}$.

\end{definition}

Operationally, we can understand condition (i) as follows.
After the encoding $\mc{E}_{\D\Q\to\bm{\B}}$, the programmer Alice chooses her desired access structure $\Gamma_i$ and performs the measurement $\{M_{a|i}\}_a$.  
Given outcome $a$, the joint space $\mc{H}_{\mathsf{A}}\otimes\mc{H}_{\bm{\B}}$ is projected into the subspace of $\mc{H}_{\bm{\B}}$ specified by Eq. \eqref{Eq:PAS}.
The procedure defines a PAS if this subspace is a QSS with access structure $\Gamma_i$.
Alice then reveals the code by announcing the measurement outcome $a$ along with her choice $i$, and authorized subsets within $\Gamma_i$ can recover the secret originally stored in $\mc{H}_\mathsf{Q}$.
Notice that Alice's measurement can occur at any point in time before or after the encoder $\mc{E}$ is applied and the shares are distributed to parties $\bm{\B}$.
Thus, the programming of the access structure can be done asynchronously to the initial reception of the quantum secret in $\mathsf{Q}$ and the final delivery of the shares $\B_i$.

Condition (ii) is natural to impose for consistency.
By monotonicity in access structures, we have that $\{\bm{\B}\}\subset\Gamma_i$ for every access structure $\Gamma_i$.
Hence, the total cohort of parties $\bm{\B}$ should be authorized to recover the secret without needing to wait for the information $a$ and $i$.
Furthermore, similar to the definition of a blind helper code, (ii) enforces a power check on the programmer so that the parties $\bm{\B}$ can still work together to recover the secret independently.

\subsection{The role of quantum steering in PASs}

Before we discuss a method for building PASs from blind helper codes, we explain an interesting connection to the nonlocal effect of quantum steering.
Observe that we can equivalently express Eq. \eqref{Eq:PAS} as
\begin{align}
\label{Eq:BPAS-code-assemblage}
    \mc{V}_{a|i}= \text{supp}\left\{ \mc{E}_{\mathsf{DQ}\to\bm{\B}} (\sigma_{a|i}\otimes\mathbb{I}_{\mathsf{Q}})\right\},
\end{align}
where
\begin{equation}
\label{Eq:state-assemblage}
    \sigma_{a|i}=\tr_{\A}[(M_{a|i}\otimes\mathbb{I}_\D)\psi_{\A\D}].
\end{equation}
Eqns. \eqref{Eq:BPAS-code-assemblage} and \eqref{Eq:state-assemblage} reveal an alternative interpretation of how the programming works.
After Alice decides on the particular access structure she wishes to impose, she measures her subsystem $\A$ in a judiciously chosen basis. 
This effectively steers the dealer's subsystem $\D$ into the state assemblage $\{\sigma_{a|i}\}_{a,i}$ given in Eq. \eqref{Eq:state-assemblage}.

State assemblages are the basic ingredient underlying the notion of quantum steering.
Formally, we say that Alice has steered the dealer's system using entangled state $\psi_{\A\D}$ if the post-measurement state assemblage $\{\sigma_{a|i}\}$ given by Eq. \eqref{Eq:state-assemblage} does not satisfy a local hidden state (LHS) model \cite{Wiseman-2007a}.
The latter means that there does \textit{not} exist a fixed sub-normalized state ensemble $\{\rho_\lambda\}_\lambda$ (satisfying $\sum_\lambda\tr[\rho_\lambda]=1$) and classical post-processing channel $p(a|i,\lambda)$ such that
\begin{align}
\label{Eq:LHS}
    \sigma_{a|i}=\sum_\lambda p(a|i,\lambda)\rho_\lambda,\quad\forall a,i.
\end{align}
The following proposition shows that steerability is a necessary condition of all PAS implementations with $|\mathfrak{A}|> 1$.
\begin{proposition}
\label{Prop:steering}
    Let $\mathfrak{A}=\{\Gamma_i\}$ contain at least two distinct access structures.  Then any PAS that realizes the programming of $\mathfrak{A}$ requires steering of the dealer's system using the shared state $\psi_{\A\D}$.
\end{proposition}
\begin{proof}
    Suppose on the contrary that the state assemblage $\{\sigma_{a|i}\}$ generated by Alice's measurement satisfies a LHS model.
    Then substituting Eq. \eqref{Eq:LHS} into Eq. \eqref{Eq:BPAS-code-assemblage} yields
    \begin{equation}
        \mc{V}_{a|i}=\text{supp}\left\{\sum_\lambda p(a|i,\lambda)  \mc{E}_{\mathsf{DQ}\to\bm{\B}} (\rho_\lambda\otimes\mathbb{I}_{\mathsf{Q}}) \right\}.
    \end{equation}
    Now for any $\rho_\lambda$ and any distinct $i,i'$, we must be able to find outcomes $\overline{a},\overline{a}'$ such that both $p(\overline{a}|i,\lambda)$ and $p(\overline{a}'|i',\lambda)$ are nonzero.
    Since $\supp(A)\subset \supp(A+B)$ for any nonzero positive operators $A$ and $B$, we have 
    \begin{align}
        \text{supp}\left\{\mc{E}_{\mathsf{DQ}\to\bm{\B}} (\rho_\lambda\otimes\mathbb{I}_{\mathsf{Q}}) \right\}\subset\mc{V}_{\overline{a}|i}\cap\mc{V}_{\overline{a}'|i'}\not=\emptyset.
    \end{align}
    But then codes $\mc{V}_{\overline{a}|i}$ and $\mc{V}_{\overline{a}'|i'}$  have the same access structure because they both must allow for the recovery of whatever is encoded in $\mc{V}_{\overline{a}|i}\cap\mc{V}_{\overline{a}'|i'}$ by any authorized set.
    Finally, the $\{\mc{V}_{a|i}\}_a$ must have the same access structure for every $a$ according to the definition of a PAS and likewise for the $\{\mc{V}_{a'|i'}\}_{a'}$.
    Hence $\Gamma_i=\Gamma_{i'}$, which contradicts the assumption $|\mc{A}|>1$.
\end{proof}

\subsection{Building PASs from helper codes}

One straightforward method to construct a PAS is to perform code concatenation starting with a blind QSS helper code and appending to it QSS codes with different access structures.
The method also requires that the helper provides assistance through classical communication alone.
Example \ref{Ex:universal-blind} shows that this latter constraint is not overly restrictive.
Indeed, by using additional $O(|\Q|+n)$ ebits, any QSS helper code can be transformed into a blind code with classical assistance, and so the method described here can also be applied to general helper codes. 

\begin{figure}[t]
\centering
\includegraphics[width=0.45\textwidth]{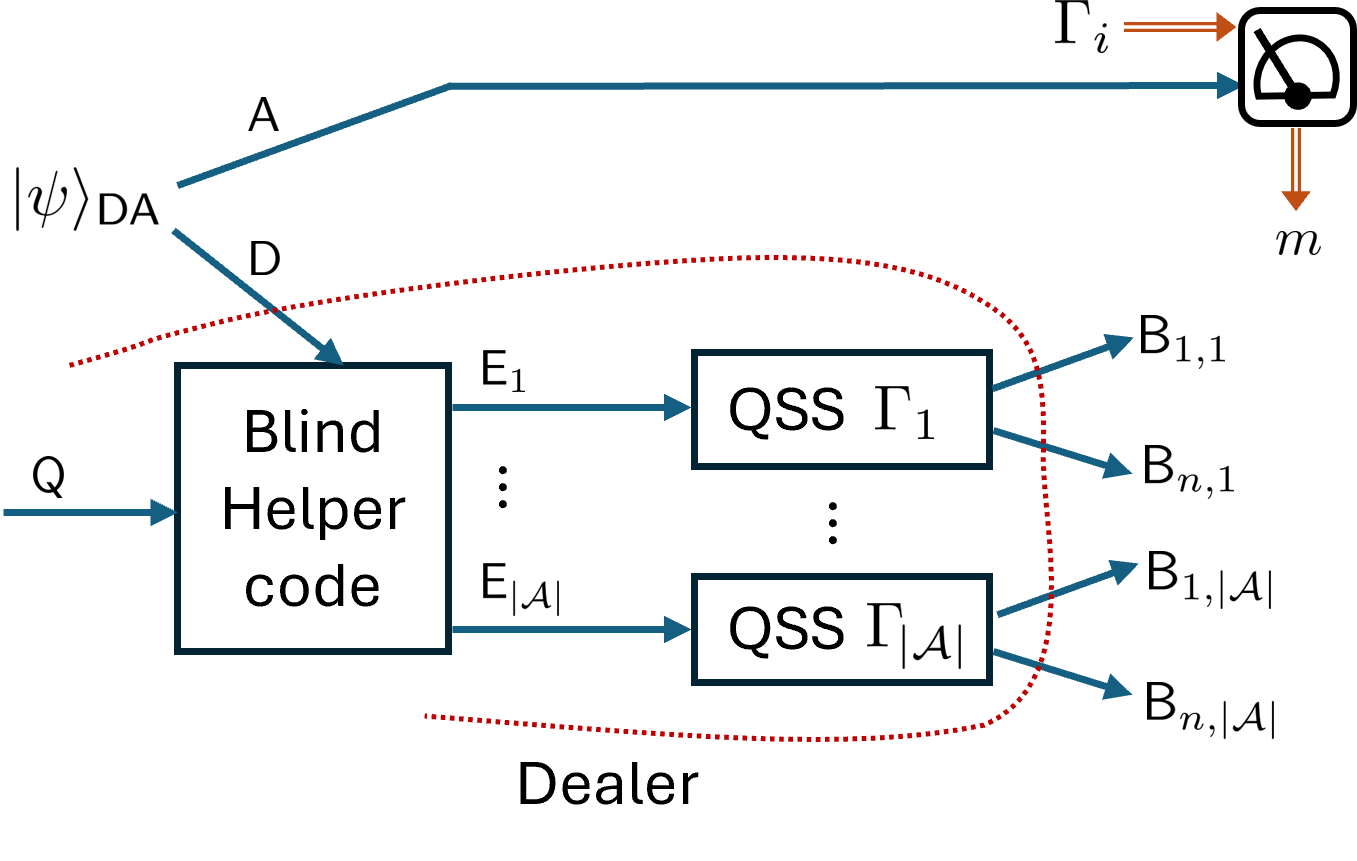}
\caption{A method for building a programmable access structure using a blind helper code.}
\label{Fig:PAS_Blind_concat}
\end{figure}

Figure \ref{Fig:PAS_Blind_concat} depicts the construction with the details given as follows.
Let $\mathfrak{A}=\{\Gamma_i\}$ be any family of valid access structures for parties $\bm{\B}$.
For arbitrary system $\Q$ and any number of parties $\bm{\E}:=\E_1,\E_2,\cdots,\E_{|\mathfrak{A}|}$, choose a blind QSS helper code that encodes $\Q$ into $\A\bm{\E}$.
We can always find such a code using, for example, the teleportation or CSS helper codes 
introduced in Section \ref{Sect:Helper-examples}.
Moreover, since the code is blind, there exists an entangled state $\ket{\psi}_{\A\D}$ such that the encoding has the form of Eq. \eqref{Eq:blind-encoding}, i.e.
\begin{align}
    \ket{i}_{\Q}\mapsto\ket{i}_{\A\bm{\E}}=U_{\D\Q\to\bm{\E}}\ket{\psi}_{\A\D}\ket{i}_{\Q}.
\end{align}
The isometry $U_{\D\Q\to\bm{\E}}$ is performed by the dealer, and it defines the first layer of the encoding map for the PAS.
The state $\ket{\psi}_{\A\D}$ is pre-shared between the programmer and the dealer in the PAS.

After the dealer applies $U$, the second layer of the code is implemented. 
On each pair of systems $\E_i\E_i'$, the dealer performs an encoding map $W_i:\mc{H}_{\E_i}\to\mc{H}_{\B_{1,i}}\otimes\mc{H}_{\B_{2,i}}\otimes\cdots\otimes\mc{H}_{\B_{n,i}}$ of any QSS code for $n$ parties with access structure $\Gamma_i$.
The $k^{th}$ share, denoted by $\B_{k,i}$ is then delivered to party $\B_i$, and its size will depend on the particular QSS code used.
For example, the constructive method for realizing a given access structure $\Gamma_i$ described in Ref. \cite{Gottesman-2000a} involves building up shares through concatenations of threshold codes.
Nevertheless, there always exists a way to encode each pair of subsystems $\E_i\E_i'$ into a QSS code having the desired access structure $\Gamma_i$.
In total then, each party $\B_i$ will receive $|\mathfrak{A}|$ subsystems from the dealer, with the $k^{th}$ subsystem $\B_{k,i}$ being a share of system $\E_i\E_i'$.
We let $\mc{H}_{\bm{\B}}:=\otimes_{k=1}^{n}\otimes_{i=1}^{|\mathfrak{A}|}\mc{H}_{\B_{k,i}}$ and $U_{\text{tot}}:=\otimes_{i=1}^{|\mathfrak{A}|}W_i U$ denote the total encoding for the PAS.

The measurement assemblage that Alice performs to achieve the programming of $\mathfrak{A}$ consists of whatever local measurement she needs to perform in the first layer QSS code (recall we required that this code have a helper that communicates classically).
Specifically, for any access structure $\Gamma_i\in\mathfrak{A}$, she performs the measurement such that the secret can be recovered on systems $\E_i$ given her measurement outcome.
The second layer of encoding $W_i$ distributes the information in $\E_i$ to the different parties $\bm{\B}_i$ while having the desired access structure $\Gamma_i$.
Any authorized set $S\in\Gamma_i$ can recover the information in $\E_i$ and therefore also the original information in $\Q$ given the measurement outcome of the helper.
Since the first layer of the code is a valid helper code, systems $\{\E_j\}_{j\not=i}$ contain no information about $\Q$.
Therefore, all the shares of these systems distributed at the second layer of the code will also contain no information about $\Q$.
Consequently, even when accounting for all the shares of $\{\E_j\}_{j\not=i}$, the access structure for recovering $\Q$ is given by $\Gamma_i$, as desired.

\section{Conclusion}
\label{Sect:Conclusion}

We have introduced helper codes as an experimentally-motivated family of quantum secret sharing codes.
These codes have highly asymmetric access structures, allowing any party to fully recover the encoded information with the assistance of a designated helper.
Such codes may be particularly useful for error correction in hybrid quantum architectures whose likelihood of experiencing erasure depends on the respective underlying physical system.
Given this physical motivation, it would be desirable to obtain more examples of helper codes and identify their structural properties.
We have pursued such a direction in a companion paper to this work \cite{Companion-temp}, and in general we anticipate that interesting connections can be found between the structure of helper codes and features of multipartite entanglement \cite{Scott-2004a}.





Considering a helper scheme for specific applications introduces additional considerations that may be the subject of future work.
For blind helper codes, the helper holds no local information of the secret prior to decoding.
However, if the helper works jointly with a party $\B_i$ to recover the secret, then ultimately some information could be extracted by the helper.
One way to prohibit this information leakage is by using codes in which the helper provides assistance only through one-way classical communication.
While we have shown how every helper code can be converted into one having this form, this generally requires additional Bell states.
It would be very useful to have a better understanding of which helper codes themselves can be decoded with one-way local operations and classical communication (LOCC) from the helper, without needing the teleportation bootstrapping of Example \ref{Ex:universal-blind}.



Finally, we would like to highlight the role that quantum steering plays in our programmable access structure protocol.
It is not surprising that some sort of nonlocal effect is needed for a remote party to control the access structure of a quantum code.
Effectively, the programmer steers the other parties into an assemblage that has the desired access structure.
To our knowledge, this is a novel use-case of quantum steering that is distinct from other previously observed connections between steering and secret sharing \cite{Xiang-2017a, Uola-2020a}. 
As future work, it would be interesting to investigate whether every steerable state can be used to gain some nonclassical advantage in a programmable secret sharing sharing scenario.
Such a result would represent a converse to our Proposition \ref{Prop:steering}.


\bigskip

\acknowledgments
E.C. thanks Alex May for helpful discussions on dimensionality bounds in secret sharing.
S.H. and M.N. were partially supported by NSF Grants 2016136 and 2112890.
D.W.K. was partially supported by NSERC Discovery Grant RGPIN-2024-400160.

\bibliography{refs}
\bibliographystyle{unsrt}

\end{document}